\documentclass[12pt,a4paper]{article}
\usepackage[latin1]{inputenc}
\usepackage{amsmath}
\usepackage{amsfonts}
\usepackage{amssymb}
\usepackage{amsthm}
\usepackage{makeidx}
\usepackage{graphicx}
\usepackage{natbib}
\usepackage{setspace}
\usepackage{hyperref}
\usepackage[margin=1in]{geometry}

\author{Marcelo Brutti Righi\footnote{We would like to thank two anonymous reviewers for their comments, which have helped to improve our manuscript. We are also grateful for the financial support of CNPq (Brazilian Research Council) and FAPERGS (Rio Grande do Sul State Research Council).}\\\textit{Federal University of Rio Grande do Sul}\\marcelo.righi@ufrgs.br}

\title{A composition between risk and deviation measures}
\date{}

\newtheorem{Def}{Definition}[section]
\newtheorem{The}[Def]{Theorem}
\newtheorem{Pro}[Def]{Proposition}
\theoremstyle{definition}

\theoremstyle{remark}
\newtheorem{Rem}[Def]{Remark}

\begin{document}
	
	\maketitle
	\begin{abstract}
		The intuition of risk is based on two main concepts: loss and variability. In this paper, we present a composition of risk and deviation measures, which contemplate these two concepts. Based on the proposed Limitedness axiom, we prove that this resulting composition, based on properties of the two components, is a coherent risk measure. Similar results for the cases of convex and co-monotone risk measures are exposed. We also provide examples of known and new risk measures constructed under this framework in order to highlight the importance of our approach, especially the role of the Limitedness axiom. 
	\end{abstract}

	\section{Introduction}
	The intuition of risk is based on two main concepts: the possibility of a negative outcome, i.e. a loss; and the variability in terms of an expected result, i.e. a deviation. Since the time when the modern theory of finance was accepted, the role of risk measurement has attracted attention. Initially, it was predominantly used as a dispersion measure, such as variance, which contemplates the second pillar of the intuition. More recently, the occurrence of critical events has turned the attention to tail-risk measurement, as is the case of well-known Value at Risk (VaR) and Expected Shortfall (ES) measures, which contemplate the first pillar. Theoretical and mathematical discussions have gained attention in the literature, giving importance to distinct axiomatic structures for classes of risk measures and their properties. See \cite{Follmer2015} for a recent review. Despite their fundamental importance, such classes present a very wide range for those risk measures that can be understood as valid or useful. Thus, they can be considered as a first step, in which measures with poor theoretical properties are discarded. The next step would be to consider, inside a class, those measures more suited to practical use. Thus, to ensure a more complete measurement, it is reasonable to consider contemplating both pillars of intuition on risk. These pillars include the possibility of negative results and variability over an expected result, as a single measure. 
	
	Some authors have proposed and studied specific examples of risk measures of this kind. \cite{Ogryczak1999} analyzed properties from the mean plus semi-deviation. \cite{Fischer2003} and \cite{Chen2008} considered combining the mean and semi-deviations at different powers to form a coherent risk measure. \cite{Furman2006} proposed a measure that weights the mean and standard deviation in the truncated tail by VaR. \cite{Krokhmal2007} and \cite{Dentcheva2010} extended the ES concept, obtained as the solution to an optimization problem, for cases with higher moments and establish a relationship to deviation measures. \cite{Righi2016} considered penalizing the ES by the dispersion of results that represent losses exceeding the ES. \cite{Furman2017} and \cite{Berkhouch2017} penalize ES by the dispersion of tail-based Gini measures. These risk measures are individual examples, rather than a general approach. The difficulty in combining both concepts arises from the loss of an individual components theoretical properties, especially the fundamental Monotonicity axiom. This property guarantees that positions with the worst outcomes have larger values for risk measures. As an example, this axiom is not respected by the very intuitive mean plus standard deviation measure, despite the very good characteristics and intuitive separate meaning of both the mean and standard deviation. 
	
	Seeking to address this deficiency, our objective in this paper is to combine risk and deviation measures in a general fashion with formulation $\rho+\mathcal{D}$ in order to maintain desired theoretical properties, which are central to the theory of risk measures. In our main context, $\rho$ is a coherent risk measure in the sense of \cite{Artzner1999}, whereas $\mathcal{D}$ is a generalized deviation measure, as proposed by \cite{Rockafellar2006}. The financial interpretation of $\rho+\mathcal{D}$ is the same as any coherent risk measure, but it serves as a more conservative protection once it yields higher values due o the penalty resulting from dispersion, while keeping the desired properties. Nevertheless, instead of making this protection arbitrary, our approach contemplates a deviation term and leads to desired theoretical properties. 
	
	We prove a useful result that relates Limitedness; an axiom we propose of the form $\rho(X)\leq-\inf(X)$, with Monotonicity and Lower Range Dominance. The milestone is that, in these cases, we always obtain $\mathcal{D}(X)\leq-\rho(X)-\inf X$, i.e. the dispersion term considers 'financial information' from the interval between the loss represented by $\rho$ and the maximum loss $-\inf X=\sup-X$. Thus, we can state that this combination is again a coherent risk measure. Under Translation Invariance, one can think in $\rho(X)+\mathcal{D}(X)$ as $\rho(X')$, where $X'=X-\mathcal{D}(X)$, i.e. a real valued penalization on the initial position $X$. Moreover, this can be extended to acceptance sets, which are composed by positions $X$ with non-positive risk, of the form  $\mathcal{A}_{\rho+\mathcal{D}}:=\{X : \rho(X)+\mathcal{D}(X)\leq 0\}=\{X : \rho(X)\leq -\mathcal{D}(X)\}$. In this sense, it is possible to explicitly observe the penalization reasoning in terms of the deviation term. A position must have risk, in terms of the loss measure $\rho$, at most of $-\mathcal{D}(X)\leq 0$ in order to be acceptable -- an even more restrictive criteria. It is valid to note, however, that although $X'=X-\mathcal{D}(X)$ works as a penalization, $\mathcal{A}_{\rho+\mathcal{D}}$ is not an acceptance set without Limitedness for $\rho+\mathcal{D}$ because Monotonicity plays a key role. 
	
	Moreover, we also discuss issues regarding Law Invariance and representations introduced in \cite{Kusuoka2001}. Our results can be extended to the case of convex measures in the sense of \cite{Follmer2002}, \cite{Frittelli2002} and \cite{Pflug2006}, or co-monotone coherent measures, as for the spectral or distortion classes proposed by \cite{Acerbi2002a} and \cite{Grechuk2009}. We also provide some examples of known and new proposed functionals composed by risk and deviation measures in order to illustrate our results, especially the role of our Limitedness axiom. In these examples, it is possible to generate the deviation term from a chosen risk measure, which eases the financial meaning. It is valid to point out that, for practical matters, both $\rho$ and $\mathcal{D}$ will be in the same monetary unit, but our results are valid even if this is not the case. Moreover, we are concerned over how to make a composition between risk and deviation measures rather than to claim it as new classes of risk measures. We highlight that, beyond the specific examples we expose, any combination of risk and deviation measures leading to Limitedness can be taken into consideration under the results we present in this paper. Moreover, our approach is static and univariate, which is standard in risk measures theory. Extensions to dynamic and multivariate cases are beyond our scope. Furthermore, extensions to a robust framework induced by uncertainty on models, for risk forecasting as in \cite{Righi2015}, linked to probability measures are also beyond the present scope.  
	
	We contribute to existing literature because, to the best of our knowledge, no such result as that proposed by us has been considered in previous studies. \cite{Rockafellar2006} presented an interplay between coherent risk measures and generalized deviation measures, and \cite{Rockafellar2013} proposed a risk quadrangle, where this relationship is extended by adding intersections with concepts of error and regret under a generator statistic. In fact, these authors prove that any given  generalized deviation $\mathcal{D}$ with $\mathcal{D}\leq E[X]-\inf X$, one can obtain the coherent risk measure $E[-X]+\mathcal{D}(X)$. However, these studies are centered on an interplay of concepts, rather than a combination that joins both pillars of the intuition on risk, since their formulation is only valid, in our notation, for the case $\rho(X)=E[-X]$. \cite{Filipovic2007} presented results in which convex functions possess Monotonicity and Translation Invariance, both of which are convex risk measures. Nonetheless, their result is based on the supremum of functions on a vector space, and not on a relation of axioms for risk measures such as in our approach.  Furthermore, we also present and prove results about some new examples of risk measures that rely on our approach.
	
	The reader should notice that the goal is to compose a new functional from risk and deviation terms, instead of decomposing a given functional into these two components. The key point is to simultaneously consider both concepts (risk and deviation) in a single functional and to guarantee the presence of theoretical properties. The approach for actuarial science, a sum of expectation plus a risk loading, does not necessarily guarantee the theoretical properties, as is the case for mean plus standard deviation, for instance. Other possibilities beyond the linear sum of risk and deviation, such as a risk measure with more risk aversion to mean-preserving spreads, may be understood as another measurement of the risk term and do not explicit the dispersion term and do not guarantee theoretical properties. It is more a dominance stochastic approach, which is related to probability distributions. 
	
	The remainder of this paper is structured as follows: Section \ref{sec:prel} presents the notation, definitions and preliminaries from the literature; Section \ref{sec:main} contains our main results regarding the proposed composition of risk and deviation measures under the Limitedness axiom; Section \ref{sec:exa} exposes examples and results of known and new proposed compositions in order to illustrate our approach, especially the role for Limitedness axiom; and Section \ref{sec:conc} summarizes and concludes the paper.

	\section{Preliminaries}\label{sec:prel}
	
	Unless otherwise stated, the content is based on the following notation. Consider the random result $X$ of any asset ($X\geq0$ is a gain, $X<0$ is a loss) that is defined in an atom-less probability space $(\Omega,\mathcal{F},\mathbb{P})$. In addition, $\mathcal{P}=\{\mathbb{Q} : \mathbb{Q}\ll\mathbb{P}\}$ is the non-empty set of probability measures $\mathbb{Q}$ defined in $(\Omega,\mathcal{F})$, which are absolutely continuous in relation to $\mathbb{P}$. We have that $\frac{d\mathbb{Q}}{d\mathbb{P}}$ is the density of $\mathbb{Q}$ relative to $\mathbb{P}$, which is known as the Radon-Nikodym derivative. $\mathcal{P}_{(0,1]}$ is the set of probability measures defined in $(0,1]$. All equalities and inequalities are considered to be almost surely in $\mathbb{P}$. $E_{\mathbb{P}}[X]$ is the expected value of $X$ under $\mathbb{P}$. $F_{X}$ is the probability function of $X$ and its inverse is $F_{X}^{-1}$, defined as $F_X^{-1}(\alpha)=\inf\{x:F_X(x)\geq\alpha\}$. We define $X^+=\max(X,0)$ and $X^-=\max(-X,0)$. Let $L^p=L^p(\Omega,\mathcal{F},\mathbb{P})$, with $1\leq p\leq \infty$, be the space of equivalence classes of random variables defined by the norm $\lVert X\rVert_{p}=(E_\mathbb{P} [|X|^p])^{\frac{1}{p}}$ with finite $p$ and $\lVert X\rVert_{\infty}=\inf\{k : |X|\leq k\}$. $X\in L^{p}$ indicates that $\lVert X\rVert_{p} < \infty$. We have that $L^{q}, \frac{1}{p}+\frac{1}{q}=1$, is the dual space of $L^{p}$. 
	
	In this section, we present some definitions and results from the literature that serve as a background to our main results. In this sense, we begin by defining the axioms for risk and deviation measures. There is a large number of possible properties. We focus on those that are most prominent in the literature and that are used in this paper. Each class of risk measures is based on a specific set of axioms. We also define the classes of risk measures that are representative in this paper.
	
	\begin{Def}
		A functional $\rho : L^{p}\rightarrow \mathbb{R}\cup\{\infty\}$ is a risk measure, which may fulfill the following properties:
		
		\begin{itemize}
			\item Monotonicity: if $X \leq Y$, then $\rho(X) \geq \rho(Y),\forall X,Y\in L^{p}$.
			\item Translation Invariance: $\rho(X+C)=\rho(X)-C,\forall X\in L^{p}, C \in \mathbb{R}$.
			\item Sub-additivity: $\rho(X+Y)\leq \rho(X)+\rho(Y),\forall X,Y\in L^{p}$.
			\item Positive Homogeneity: $\rho(\lambda X)=\lambda \rho(X),\forall X\in L^{p},\lambda \geq 0$.  
			\item Convexity: $\rho(\lambda X+(1-\lambda)Y)\leq \lambda \rho(X)+(1-\lambda)\rho(Y),\forall X,Y\in L^{p},0 \leq\lambda \leq 1$.
			\item Fatou Continuity: if $\arrowvert X_{n}\arrowvert\leq Y, \{X_{n}\}_{n=1}^{\infty},Y\in L^{p}$, and $X_{n}\rightarrow X$, then $\rho(X) \leq \liminf \rho( X_{n})$. 
			\item Law Invariance: if $F_{X}=F_{Y}$, then $\rho(X)=\rho(Y),\forall X,Y\in L^{p}$. 
			\item Co-monotonic Additivity: $\rho(X+Y)= \rho(X)+\rho(Y),\forall X,Y\in L^{p}$ with $X,Y$ co-monotone, i.e. $\left( X(w)-X(w^{'})\right)\left( Y(w)-Y(w^{'}) \right)\geq0,\:\forall w,w^{'}\in\Omega$. 
			\item Limitedness: $\rho(X)\leq -\inf X=\sup -X,\forall X\in L^{p}$.
		\end{itemize}
	\end{Def}
	
	\begin{Rem}
		Monotonicity requires that, if one position generates worse results to another, its risk shall be greater. Translation Invariance ensures that, if a certain gain is added to a position, its risk shall decrease by the same amount. Risk measures that satisfy both Monotonicity and Translation Invariance are called monetary and are Lipschitz continuous in $L^\infty$. Sub-additivity, which is based on the principle of diversification, implies that the risk of a combined position is less than the sum of individual risks. Positive Homogeneity is related to the position size, i.e. the risk proportionally increases with position size. These two axioms are together known as sub-linearity. Convexity is a well-known property of functions that can be understood as a relaxed version of sub-linearity. Any two axioms among Positive Homogeneity, Sub-Additiviy and Convexity implies the third one. Fatou continuity is a well-established property for functions, directly linked to lower semi-continuity and continuity from above. Law invariance ensures that two positions with the same probability function have equal risks. Co-monotonic Additivity is an extreme case where there is no diversification, because the positions have perfect positive association. Co-monotonic Additivity implies Positive Homogeneity. Limitedness ensures that the risk of a position is never greater than the maximum loss. In this paper, we are always working with normalized risk measures in the sense of $\rho(0)=0$, since this is easily obtained through a translation. 
	\end{Rem}
	
	\begin{Def}
		A functional $\mathcal{D} : L^{p}\rightarrow\mathbb{R}_{+}\cup\{\infty\}$ is a deviation measure, which may fulfill the following properties:
		\begin{itemize}
			\item Non-Negativity: For all $X\in L^p$, $\mathcal{D}(X)=0$ for constant $X$, and $\mathcal{D}(X)>0$ for
			non-constant X.
			\item Translation Insensitivity: $\mathcal{D}(X+C)=\mathcal{D}(X),\forall X\in L^{p}, C \in \mathbb{R}$.
			\item Sub-additivity: $\mathcal{D}(X+Y)\leq \mathcal{D}(X)+\mathcal{D}(Y),\forall X,Y\in L^{p}$.
			\item Positive Homogeneity: $\mathcal{D}(\lambda X)=\lambda \mathcal{D}(X),\forall X\in L^{p},\lambda \geq 0$.
			\item Lower Range Dominance: $\mathcal{D}(X)\leq E_\mathbb{P}[X]-\inf X,\forall X\in L^{p}$.
			\item Fatou Continuity: if $\arrowvert X_{n}\arrowvert\leq Y, \{X_{n}\}_{n=1}^{\infty},Y\in L^{p}$, and $X_{n}\rightarrow X$, then $\mathcal{D}(X) \leq \liminf \mathcal{D}( X_{n})$. 
			\item Law Invariance: if $F_{X}=F_{Y}$, then $\mathcal{D}(X)=\mathcal{D}(Y),\forall X,Y\in L^{p}$. 
			\item Co-monotonic Additivity: $\mathcal{D}(X+Y)= \mathcal{D}(X)+\mathcal{D}(Y),\forall X,Y\in L^{p}$ with $X,Y$ co-monotone. 
		\end{itemize}
	\end{Def}
	
	\begin{Rem}
		Non-negativity assures that there is dispersion only for non-constant positions. Translation Insensitivity indicates that the deviation does not change if a constant value is added. Lower Range Dominance restricts the measure to a range that is lower than the range between expectation and the minimum value. These axioms are related to the concept of norm, as explored in \cite{Righi2017}. 
	\end{Rem}
	
	\begin{Def}
		Let $\rho : L^{p}\rightarrow\mathbb{R}\cup\{\infty\}$ and $\mathcal{D} : L^{p}\rightarrow\mathbb{R}_{+}\cup\{\infty\}$.
		
		\begin{enumerate}
			\item $\rho$ is a coherent risk measure in the sense of \cite{Artzner1999} if it fulfills the axioms of Monotonicity, Translation Invariance, Sub-additivity, and Positive Homogeneity. 
			
			\item $\rho$ is a convex risk measure in the sense of \cite{Follmer2002} and \cite{Frittelli2002} if it fulfills the axioms of Monotonicity, Translation Invariance, and Convexity. 
			
			\item $\mathcal{D}$ is a generalized deviation measure in the sense of \cite{Rockafellar2006} if it fulfills the axioms of Non-negativity, Translation Insensitivity, Sub-additivity, and Positive Homogeneity.
			
			\item $\mathcal{D}$ is a convex deviation measure in the sense of \cite{Pflug2006} if it fulfills the axioms of Non-negativity, Translation Insensitivity, and Convexity.
			
			\item A risk or deviation measure is said to be law invariant, lower-range dominated, limited, co-monotone, or Fatou continuous if it fulfills the axioms of Law Invariance, Lower Range Dominance, Limitedness, Co-monotonic Additivity, or Fatou Continuity, respectively.
		\end{enumerate}
	\end{Def}
	
	\begin{Rem}
		Given a coherent risk measure $\rho$, it is possible to define an acceptance set of positions that do not have positive risk as $\mathcal{A}_{\rho}=\{X\in L^{p} : \rho(X)\leq 0\}$ . Let $L^{p}_{+}$ be the cone of the non-negative elements of $L^{p}$ and $L^{p}_{-}$ its negative counterpart. This acceptance set contains $L^{p}_{+}$, has no intersection with $L^{p}_{-}$, and is a convex cone. The risk measure associated with this set is $\rho(X)=\inf\{m : X+m\in \mathcal{A}_{\rho}\}$, i.e. the minimum capital that needs to be added to $X$ to ensure it becomes acceptable. For convex risk measures, $\mathcal{A}_{\rho}$ need not be a cone.
	\end{Rem}
	
	A coherent risk measure can be represented as the worst possible expectation from scenarios generated by probability measures $\mathbb{Q}\in\mathcal{P}$, known as dual sets. \cite{Artzner1999} presented this result for finite $L^{\infty}$ spaces. \cite{Delbaen2002} generalized for all $L^{\infty}$ spaces, whereas \cite{Inoue2003} considered the spaces $L^{p}, 1\leq p\leq\infty$. \cite{Follmer2002}, \cite{Frittelli2002} and \cite{Kaina2009} presented a similar result for convex risk measures based on a penalty function. It is also possible to represent generalized deviation measures in a similar approach, with the due adjustments, as demonstrated by \cite{Rockafellar2006} and \cite{Grechuk2009}. \cite{Ang2018} adapted this framework for coherent risk measures. \cite{Pflug2006} proved similar results for convex deviation measures also based on a penalty function. In this sense, the dual representations we consider in this paper are formally guaranteed by the following results. 
	
	\begin{The}\label{the:dua}
		Let $\rho : L^{p}\rightarrow \mathbb{R}\cup\{\infty\}$ and $\mathcal{D} : L^{p}\rightarrow\mathbb{R}_{+}\cup\{\infty\}$. Then:
		\begin{enumerate}
			\item $\rho$ is a Fatou continuous coherent risk measure if, and only if, it can be represented as $\rho(X)=\sup\limits_{\mathbb{Q}\in\mathcal{P}_{\rho}} E_{\mathbb{Q}} [-X]$, where $\mathcal{P}_{\rho}=\{\mathbb{Q}\in\mathcal{P} : \frac{d\mathbb{Q}}{d\mathbb{P}}\in L^{q}, \rho(X)\geq E_{\mathbb{Q}}[-X],\forall X\in L^{p}\}$ is a closed and convex dual set.
			\item  $\rho$ is a Fatou continuous convex risk measure if, and only if, it can be represented as $\rho(X)=\sup\limits_{\mathbb{Q}\in\mathcal{P}} \{ E_{\mathbb{Q}} [-X]-\gamma_{\rho}(\mathbb{Q})\}$, where $\gamma_{\rho} : \mathcal{P}\rightarrow\mathbb{R}\cup\{\infty\}$ is a lower semi-continuous convex penalty function conform $\gamma_{\rho}(\mathbb{Q})=\sup\limits_{X\in\mathcal{A}_{\rho}} E_{\mathbb{Q}} [-X]$, with $\gamma_{\rho}(\mathbb{Q})\geq-\rho(0)$. 
			\item $\mathcal{D}$ is a lower-range dominated Fatou continuous generalized deviation measure if, and only if, it can be represented as $\mathcal{D}(X)=E_{\mathbb{P}} [X]-\inf\limits_{\mathbb{Q\in\mathcal{P}_{\mathcal{D}}}}E_{\mathbb{Q}} [X]$, where $\mathcal{P}_{\mathcal{D}}=\{\mathbb{Q}\in\mathcal{P} : \frac{d\mathbb{Q}}{d\mathbb{P}}\in L^{q},\mathcal{D}(X)\geq E_{\mathbb{P}}[X]-E_{\mathbb{Q}}[X],\forall X\in L^{p}\}$ is a closed and convex dual set. 
			\item $\mathcal{D}$ is a lower-range dominated Fatou continuous convex deviation measure if, and only if, it can be represented as $\mathcal{D}(X)=E_{\mathbb{P}} [X]-\inf\limits_{\mathbb{Q\in\mathcal{P}}}\{ E_{\mathbb{Q}} [X]+\gamma_{\mathcal{D}}(\mathbb{Q})\}$, where $\gamma_{\mathcal{D}}$ is similar to $\gamma_{\rho}$.  
		\end{enumerate}
	\end{The}
	
	\section{Main results}\label{sec:main}

	We now turn the focus to our main contribution, the proposed approach for combination of risk and deviation measures. We initially prove interesting results that relate Monotonicity and Lower Range Dominance axioms to Limitedness. Based on these and the previously exposed results, we are able to prove our main theorem. The results can be extended to the convex and co-monotone coherent cases. 
	
	\begin{Pro}\label{pro:lim}
		Let $\rho : L^{p}\rightarrow\mathbb{R}\cup\{\infty\}$ and $\mathcal{D} : L^{p}\rightarrow\mathbb{R}_{+}\cup\{\infty\}$.Then:
		
		\begin{enumerate}
			\item If $\rho$ fulfills Sub-additivity (Convexity) and Limitedness, then it possesses Monotonicity. 
			\item If $\rho$ fulfills Translation Invariance and Monotonicity, then it possesses Limitedness.
			\item if $\rho+\mathcal{D}$ is a coherent (convex) risk measure, then $\mathcal{D}$ possesses Lower Range Dominance.  
		\end{enumerate}
		
	\end{Pro}
	\begin{proof}
		For \textit{(i)}, remember that because $L^p$ spaces are composed by equivalence classes of random variables, we have that if $X = Y$, then $\rho(X)= \rho(Y)$. We begin by supposing the Sub-additivity of $\rho$. Let $X,Y\in L^{p},X\leq Y$. There is $Z\in L^{p},Z\geq0$ such as $Y=X+Z$. From Limitedness, we must have $\rho(Z)\leq -\inf Z\leq0$. Thus, by Sub-additivity we obtain $\rho(Y)=\rho(X+Z)\leq\rho(X)+\rho(Z)\leq\rho(X)$, as required. By the same logic, let $\rho$ have Convexity. Thus, for any $\lambda\in(0,1)$ there is some $Z\in L^{p},\:Z\geq0$ so that we have $Y=\lambda X+(1-\lambda)Z$. This leads to $\rho(Y)=\rho(\lambda X+(1-\lambda)Z)\leq\lambda\rho(X)+(1-\lambda)\rho(Z)\leq\lambda\rho(X)$. As $\lambda$ is an arbitrary value in $(0,1)$, we can make it as close to $1$ as we want and obtain $\rho(Y)\leq\rho(X)$, as desired. 
		
		For \textit{(ii)}, note that because $X\geq\inf X$, Monotonicity and Translation Invariance implies $\rho(X)\leq\rho(\inf X)=-\inf X$, which is Limitedness. 
		
		For \textit{(iii)}, note that for a coherent (convex) risk measure $\rho$, due to its dual representation, we have that $E_{\mathbb{P}} [-X]\leq\rho(X)\leq\sup -X=-\inf X$ with extreme situations where $\mathcal{P}_{\rho}$ equals a singleton $\{\mathbb{P}\}$ or the whole $\mathcal{P}$. Thus, if $\rho+\mathcal{D}$ is coherent (convex), hence limited, then $\mathcal{D}$ is lower-range dominated because $\mathcal{D}(X)\leq -\rho(X)-\inf X\leq E_{\mathbb{P}} [X]-\inf X$. This concludes the proof.
	\end{proof}
	
	\begin{Rem}
		As proved by \cite{Bauerle2006}, in the presence of Law Invariance, Convexity and Monotonicity are equivalent to second-order stochastic dominance for atom-less spaces. As Limitedness implies Monotonicity, in the presence of Convexity and Law Invariance, it also implies second-order stochastic dominance.
	\end{Rem}
	
	\begin{The}\label{the:main}
		Let $\rho : L^{p}\rightarrow \mathbb{R}\cup\{\infty\}$ be a coherent risk measure and $\mathcal{D} : L^{p}\rightarrow \mathbb{R}_{+}\cup\{\infty\}$ a generalized deviation measure. Then:
		
		\begin{enumerate}
			\item $\rho+\mathcal{D}$ is a coherent risk measure if and only if it fulfills Limitedness. 
			\item $\rho$ and $\mathcal{D}$ are Fatou continuous and $\rho+\mathcal{D}$ is limited if, and only if, $\rho+\mathcal{D}$ can be represented as $\rho(X)+\mathcal{D}(X)=\sup\limits_{\mathbb{Q}\in\mathcal{P}_{\rho+\mathcal{D}}} E_{\mathbb{Q}} [-X]$, where $\mathcal{P}_{\rho+\mathcal{D}}=\{\mathbb{Q}\in\mathcal{P} : \frac{d\mathbb{Q}}{d\mathbb{P}}=\frac{d\mathbb{Q}_{\rho}}{d\mathbb{P}}+\frac{d\mathbb{Q}_{\mathcal{D}}}{d\mathbb{P}}-1,\mathbb{Q}_{\rho}\in\mathcal{P}_{\rho},\mathbb{Q}_{\mathcal{D}}\in\mathcal{P}_{\mathcal{D}} \}$.
			\item $\rho$ and $\mathcal{D}$ are law invariant and $\rho+\mathcal{D}$ is limited if, and only if, $\rho+\mathcal{D}$ can be represented as $\rho(X)+\mathcal{D}(X)=\sup\limits_{m\in\mathcal{M}}\int_{0}^{1}\rho^{\alpha}(X)md(\alpha)$, where $\rho^{\alpha}(X)=-\frac{1}{\alpha}\int_{0}^{\alpha}F_{X}^{-1}(u)du$ and $\mathcal{M}=\{m\in\mathcal{P}_{(0,1]} : \int_{(0,1]} \frac{1}{\alpha}dm(\alpha)=\frac{d\mathbb{Q}}{d\mathbb{P}},\mathbb{Q}\in\mathcal{P}_{\rho+\mathcal{D}}\}$.
		\end{enumerate}	
	\end{The}
	\begin{proof}
		We begin with \textit{(i)}. According to Proposition \ref{pro:lim}, if $\rho+\mathcal{D}$ is a coherent risk measure then it fulfills Limitedness. For the converse part, the Translation Invariance, Sub-additivity, and Positive Homogeneity of $\rho+\mathcal{D}$ is a consequence of the individual axioms fulfilled by $\rho$ and $\mathcal{D}$ individually by definition. As there is Limitedness by assumption, $\rho+\mathcal{D}$ respects Monotonicity due to Proposition \ref{pro:lim}. Hence, it is a coherent risk measure.
		
		For \textit{(ii)}, $\rho+\mathcal{D}$ being limited implies it is a coherent risk measure, by the previous result. As $\rho$ and $\mathcal{D}$ are Fatou continuous, by Theorem \ref{the:dua} they have representations with dual sets $\mathcal{P}_{\rho}$ and $\mathcal{P}_{\mathcal{D}}$. Thus, $\rho+\mathcal{D}$ is also Fatou continuous and has dual representation. We then obtain that:
		\begin{align*}
			\rho(X)+\mathcal{D}(X) &=\sup\limits_{\mathbb{Q}_{\rho}\in\mathcal{P}_{\rho}} E_{\mathbb{Q}_{\rho}} [-X]+E_{\mathbb{P}} [X]-\inf\limits_{\mathbb{Q}_{\mathcal{D}}\in\mathcal{P}_{\mathcal{D}}}E_{\mathbb{Q}_{\mathcal{D}}} [X]\\
			&=\sup\limits_{\mathbb{Q}_{\rho}\in\mathcal{P}_{\rho},\mathbb{Q}_{\mathcal{D}}\in\mathcal{P}_{\mathcal{D}}} \{E_{\mathbb{Q}_{\rho}} [-X]-E_{\mathbb{P}} [-X]+E_{\mathbb{Q}_{\mathcal{D}}} [-X]\}\\
			&=\sup\limits_{\mathbb{Q}_{\rho}\in\mathcal{P}_{\rho},\mathbb{Q}_{\mathcal{D}}\in\mathcal{P}_{\mathcal{D}}} \left\lbrace E_{\mathbb{P}} \left[ -X\left( \frac{d\mathbb{Q}_{\rho}}{d\mathbb{P}}+\frac{d\mathbb{Q}_{\mathcal{D}}}{d\mathbb{P}}-1\right) \right] \right\rbrace \\
			&=\sup\limits_{\mathbb{Q}\in\mathcal{P}_{\rho+\mathcal{D}}} E_{\mathbb{Q}} [-X],
		\end{align*}
		where $\mathcal{P}_{\rho+\mathcal{D}}=\{\mathbb{Q}\in\mathcal{P} : \frac{d\mathbb{Q}}{d\mathbb{P}}=\frac{d\mathbb{Q}_{\rho}}{d\mathbb{P}}+\frac{d\mathbb{Q}_{\mathcal{D}}}{d\mathbb{P}}-1,\mathbb{Q}_{\rho}\in\mathcal{P}_{\rho},\mathbb{Q}_{\mathcal{D}}\in\mathcal{P}_{\mathcal{D}} \}$. To show that $\mathcal{P}_{\rho+\mathcal{D}}$ is composed by valid probability measures, we verify that for $\mathbb{Q}\in\mathcal{P}_{\rho+\mathcal{D}}$, $E_{\mathbb{P}}\left[\frac{d\mathbb{Q}}{d\mathbb{P}} \right] =E_{\mathbb{P}}\left[\frac{d\mathbb{Q}_{\rho}}{d\mathbb{P}} \right]+E_{\mathbb{P}}\left[\frac{d\mathbb{Q}_{\mathcal{D}}}{d\mathbb{P}} \right]-E_{\mathbb{P}}\left[1 \right]=1$. In addition, Monotonicity of $\rho+\mathcal{D}$ implies that $\frac{d\mathbb{Q}}{d\mathbb{P}}\geq0$, because for $X,Y\in L^p$ with $X\leq Y$ we have from the dual representation that $\sup\limits_{\mathbb{Q}\in\mathcal{P}_{\rho+\mathcal{D}}} E_{\mathbb{P}} [-X\frac{d\mathbb{Q}}{d\mathbb{P}}]\geq\sup\limits_{\mathbb{Q}\in\mathcal{P}_{\rho+\mathcal{D}}} E_{\mathbb{P}} [-Y\frac{d\mathbb{Q}}{d\mathbb{P}}]$. This inequality could not be guaranteed if $\frac{d\mathbb{Q}}{d\mathbb{P}}$ assume negative values. Now, we assume that $\rho+\mathcal{D}$ has such dual representation. Then $\rho+\mathcal{D}$ is a Fatou continuous coherent risk measure that respects Limitedness. Reversing the deduction steps, one recovers the individual dual representations of both $\rho$ and $\mathcal{D}$. By Theorem \ref{the:dua} these two measures possess Fatou continuity.
		
		Regarding \textit{(iii)}, \cite{Kusuoka2001} showed that coherent risk measures that fulfill Law Invariance and Fatou continuity axioms can be represented as $\sup\limits_{m\in\mathcal{M}}\int_{0}^{1}\rho^{\alpha}(X)md(\alpha)$ for some $\mathcal{M}\subset\mathcal{P}_{(0,1]}$. Moreover, \cite{Jouini2006} and \cite{Svindland2010} have proved that law-invariant convex risk measures defined in atom-less spaces are Fatou continuous. Since we are considering an atom-less probability space, we get that $\rho+\mathcal{D}$ can have this kind of representation because it is limited, then coherent. We can define a continuous variable $u\sim\mathbb{U}(0,1)$ uniformly distributed between 0 and 1, such that $F_{X}^{-1}(u)=X$. For $\mathbb{Q}\in\mathcal{P}_{\rho+\mathcal{D}}$, we can obtain $\frac{d\mathbb{Q}}{d\mathbb{P}}=H(u)=\int_{(u,1]}\frac{1}{\alpha}dm(\alpha)$, where $H$ is a monotonically decreasing function and $m\in\mathcal{P}_{(0,1]}$. As $H$ is anti-monotonic in relation to X, one can reach the supremum in a dual representation. Then we obtain 
		\begin{align*}
			\rho(X)+\mathcal{D}(X) &=\sup\limits_{\mathbb{Q}\in\mathcal{P}_{\rho+\mathcal{D}}} E_{\mathbb{Q}} [-X]\\
			&=\sup\limits_{\mathbb{Q}\in\mathcal{P}_{\rho+\mathcal{D}}} E_{\mathbb{P}} \left[ -X\frac{d\mathbb{Q}}{d\mathbb{P}}\right] \\
			&=\sup\limits_{m\in\mathcal{M}} \left\lbrace\int_{0}^{1}-F_{X}^{-1}(u) \left[ \int_{(u,1]}\frac{1}{\alpha}dm(\alpha)\right] du\right\rbrace  \\
			&=\sup\limits_{m\in\mathcal{M}} \left\lbrace\int_{(0,1]} \left[\frac{1}{\alpha} \int_{0}^{\alpha}-F_{X}^{-1}(u)du\right]dm(\alpha)\right\rbrace  \\
			&=\sup\limits_{m\in\mathcal{M}} \left\lbrace\int_{(0,1]} \rho^{\alpha}dm(\alpha)\right\rbrace,
		\end{align*}
		where $\mathcal{M}=\left\lbrace m\in\mathcal{P}_{(0,1]} : \int_{(u,1]}\frac{1}{\alpha}dm(\alpha)=\frac{d\mathbb{Q}}{d\mathbb{P}},\mathbb{Q}\in\mathcal{P}_{\rho+\mathcal{D}} \right\rbrace$. We now assume that $\rho+\mathcal{D}$ has such representation. Then it is a law-invariant coherent risk measure. This is only possible if both $\rho$ and $\mathcal{D}$ are law invariant. By \textit{(i)}, it is also limited. This concludes the proof.
	\end{proof}
	
	Assertions of Theorem \ref{the:main} can be extended to the case where $\rho$ is a convex risk measure and $\mathcal{D}$ a convex deviation measure. For the law invariant case, \cite{Frittelli2005} and \cite{Noyan2015} proved representations similar to those of \cite{Kusuoka2001} for convex risk measures. The results of Theorem \ref{the:main} can also be extended to the case where $\rho$ and $\mathcal{D}$ are co-monotone. In this scenario, $\mathcal{M}$ becomes a singleton, as is the case with the spectral risk measures proposed by \cite{Acerbi2002a} and concave distortion functions, which are widely used in insurance. \cite{Grechuk2009}, \cite{Wang2017} and \cite{Furman2017} prove results linking these classes and axioms for generalized deviation measures. We state these two extensions without proof, because the deductions are quite similar to the coherent case.
	
	\begin{The}\label{the:con}
		Let $\rho : L^{p}\rightarrow \mathbb{R}\cup\{\infty\}$ be a convex risk measure and $\mathcal{D} : L^{p}\rightarrow \mathbb{R}_{+}\cup\{\infty\}$ a convex deviation measure. Then:
		
		\begin{enumerate}
			\item $\rho+\mathcal{D}$ is a convex risk measure if and only if it fulfills Limitedness. 
			\item $\rho$ and $\mathcal{D}$ are Fatou continuous and $\rho+\mathcal{D}$ is limited if, and only if, $\rho+\mathcal{D}$ can be represented as $\rho(X)+\mathcal{D}(X)=\sup\limits_{\mathbb{Q}\in\mathcal{P}} \{ E_{\mathbb{Q}} [-X]-\gamma_{\rho+\mathcal{D}}(\mathbb{Q})\}$, where $\gamma_{\rho+\mathcal{D}}=\gamma_{\rho}+\gamma_{\mathcal{D}}$.
			\item $\rho$ and $\mathcal{D}$ are law invariant and $\rho+\mathcal{D}$ is limited if, and only if, $\rho+\mathcal{D}$ can be represented as $\rho(X)+\mathcal{D}(X)=\sup\limits_{m\in\mathcal{P}(0,1]}\left\lbrace \int_{0}^{1}\rho^{\alpha}(X)dm(\alpha)-\gamma_{\rho+\mathcal{D}}(m)\right\rbrace $. 
		\end{enumerate}	
	\end{The}
	
	\begin{The}\label{the:com}
		Let $\rho : L^{p}\rightarrow \mathbb{R}\cup\{\infty\}$ be a co-monotone coherent risk measure and $\mathcal{D} : L^{p}\rightarrow \mathbb{R}_{+}\cup\{\infty\}$ a co-monotone generalized deviation measure. Then:
		
		\begin{enumerate}
			\item $\rho+\mathcal{D}$ is a co-monotone coherent risk measure if, and only if, it fulfills Limitedness. 
			\item $\rho$ and $\mathcal{D}$ are Fatou continuous and $\rho+\mathcal{D}$ is limited if, and only if, $\rho+\mathcal{D}$ can be represented as $\rho(X)+\mathcal{D}(X)=\sup\limits_{\mathbb{Q}\in\mathcal{P}_{\rho+\mathcal{D}}} E_{\mathbb{Q}} [-X]$.
			\item $\rho$ and $\mathcal{D}$ are law invariant and $\rho+\mathcal{D}$ is limited if, and only if, $\rho+\mathcal{D}$ can be represented as $\rho(X)+\mathcal{D}(X)=\int_{0}^{1}\rho^{\alpha}(X)dm(\alpha)$, where $m\in\mathcal{P}_{(0,1]}$.
		\end{enumerate}	
	\end{The}

	\section{Examples}\label{sec:exa}

	In this section, we provide examples of functionals composed by risk and deviation measures in order to illustrate the importance of Limitedness, since it is central to our results. In practical situations, typically the idea is to consider $\rho+\beta\mathcal{D}$, where $\beta$ assumes the role of some penalty coefficient indicating that the proportion of deviation that must be included. Thus, it works similarly to an aversion term. Note that, if $\mathcal{D}$ is a generalized deviation measure, then so is $\beta\mathcal{D}$ for $\beta>0$. The same is true if $\mathcal{D}$ is convex or co-monotone generalized. It is valid to point out that in this situation the acceptance set is defined as $\mathcal{A}_{\rho+\beta\mathcal{D}}:=\{X : \rho(X)+\beta\mathcal{D}(X)\leq 0\}=\{X : -\frac{\rho(X)}{\mathcal{D}(X)}\geq \beta\}$, which is related to a performance criteria similar to a Sharpe ratio. The sing of minus is due to $\rho$ represent losses. When $\beta=0$, $\mathcal{D}$ lacks the Non-Negativity axiom. Nonetheless, in this case, our composition is trivially the initial risk measure $\rho$. We explore results with a main focus on the class of coherent risk measures, especially dual representations. Representations regarding Convexity, Law Invariance and Co-monotonic Additivity can be obtained in the same spirit as in the previous theorems.
	
	Our first example is the intuitive mean plus (p-norm) standard deviation, directly linked to the variance premium and mean-variance Markowitz portfolio theory. The negatives of mean and standard deviation are canonical examples of coherent risk and generalized deviation measures, respectively. We now define this risk measure.
	
	\begin{Def}
		The mean plus standard deviation is a functional $MSD^{\beta}:L^p\rightarrow\mathbb{R}\cup\{\infty\}$ defined conform:
		\begin{equation*}
			MSD^\beta(X)=-E_{\mathbb{P}}[X]+\beta\lVert X-E_{\mathbb{P}}[X]\rVert_p,\:0\leq\beta\leq 1.
		\end{equation*}
	\end{Def}
	
	This risk measure fulfills Translation Invariance, Convexity, Positive Homogeneity and Law Invariance. However, it does not possess Monotonicity. This fact is due to the lack of Limitedness, since it is easy to obtain $\beta\lVert X-E_{\mathbb{P}}[X]\rVert_p>E_{\mathbb{P}}[X]-\inf X$ for some random variable with skewed $F_X$. Indeed, by considering the whole distribution of $X$ makes this risk measure flawed in its financial meaning because it penalizes profit and loss in the same way. Its tail counterpart, when $X$ is restricted to values below its $\alpha$-quantile $F_X^{-1}(\alpha)$, is proposed and studied by \cite{Furman2006} and inherits its main properties. In order to circumvent such drawbacks, it becomes necessary to consider the mean plus (p-norm) semi-deviation. We give a formal definition.

	\begin{Def}
		The mean plus semi-deviation is a functional $MSD^{\beta}_{-}:L^p\rightarrow\mathbb{R}\cup\{\infty\}$ defined conform:
		\begin{equation*}
			MSD_{-}^{\beta}(X)=-E_{\mathbb{P}}[X]+\beta \lVert (X-E_{\mathbb{P}}[X])^-\rVert_p, 0\leq\beta\leq1.
		\end{equation*}
	\end{Def}
	
	It is clear that semi-deviation is a lower range dominated generalized deviation measure. This risk measure is studied in detail by \cite{Ogryczak1999} and \cite{Fischer2003}. It is well known that this functional is a law invariant coherent risk measure. We now provide an alternative proof based on our setting in order to explicit the role of Limitedness axiom.
	
	\begin{Pro}\label{Pro:MSD}
		The mean plus semi-deviation is a law invariant coherent risk measure with dual set $\mathcal{P}_{MSD^\beta_-}=\left\lbrace\mathbb{Q}\in\mathcal{P} : \frac{d\mathbb{Q}}{d\mathbb{P}}=1+\beta( W-E_{\mathbb{P}}[W]), W\leq0, \lVert W\rVert_q\leq1 \right\rbrace$.
	\end{Pro}
	\begin{proof}
		From the properties of both components, which are law invariant coherent risk and generalized deviation measures respectively, we have from Theorem \ref{the:main} the combination is a coherent risk measure if and only if it is limited. This comes from the fact that $(X-E_{\mathbb{P}}[X])^-\leq E_{\mathbb{P}}[X]-\inf X,\:\forall X\in L^p$. Thus, we have $E_{\mathbb{P}}[X]-\inf X\geq\lVert (X-E_{\mathbb{P}}[X])^-\rVert_{\infty}\geq\lVert (X-E_{\mathbb{P}}[X])^-\rVert_p\geq\beta\lVert (X-E_{\mathbb{P}}[X])^-\rVert_p$. Hence, $MSD_{-}^{\beta}\leq-\inf X$. 
		
		Regarding the structure of $\mathcal{P}_{MSD^\beta}$, one must note that the dual set of $-E_{\mathbb{P}}[X]$ is a singleton, while for the semi-deviation multiplied by $\beta$ it is composed, conform \cite{Rockafellar2006}, by relative densities of the form  $\frac{d\mathbb{Q}}{d\mathbb{P}}=\beta(1+E_{\mathbb{P}}[W]-W)+(1-\beta), W\leq0, \lVert V\rVert_q\leq1$. From Theorem \ref{the:main} we have that the representation is given by $\mathcal{P}_{MSD^\beta_-}=\left\lbrace\mathbb{Q}\in\mathcal{P} : \frac{d\mathbb{Q}}{d\mathbb{P}}=1+\beta( E_{\mathbb{P}}[W]-W), W\leq0, \lVert W\rVert_q\leq1 \right\rbrace$. This concludes the proof.
	\end{proof}
	
	From the previous proposition, we can see the importance of Limitedness for Monotonicity of the mean plus semi-deviation risk measure. This notion of penalization over a risk measure by the deviation of results worst than this value can be extended when the negative expectation is replaced by alternative risk measures. An advantage of this approach is that the agent chooses a risk measure and the deviation is directly generated from it. This is explored when $\rho$ is the ES, by \cite{Righi2016}. Moreover, \cite{Righi2017} explores for other risk measures beyond ES, such as expectiles and entropic ones, calling the approach a loss-deviation for portfolio optimization. Their results point out the advantages of these risk measures, but no theoretical results are presented. We thus present a formal definition and explore theoretical properties. In order to ease notation, we define $\rho^{*}(X)=-\rho(X)$. The minus sign is simply an adjustment to ease the notation.
	
	\begin{Def}
		Let $\rho : L^{p}\rightarrow \mathbb{R}\cup\{\infty\}$ be a risk measure. Then its loss-deviation is a functional $LD^\beta_{\rho} : L^{p}\rightarrow \mathbb{R}\cup\{\infty\}$ defined conform:
		\begin{equation*}
			LD^\beta_{\rho}(X)=\rho(X)+\beta \lVert (X-\rho^{*}(X))^-\rVert_p, 0\leq\beta\leq1.
		\end{equation*}
	\end{Def} 
	
	Despite this very interesting intuitive meaning, the penalization term $\lVert (X-\rho^{*}(X))^-\rVert_p$ is not sub-additive for any convex risk measure, with the exception of the negative expectation. To see this fact, note that $(X+Y-\rho^{*}(X)-\rho^{*}(Y))^-\leq(X+Y-\rho^{*}(X+Y))^-$, with equality if and only if $\rho(X)=-E_{\mathbb{P}}[X]$. Thus, it is not a generalized, even convex, deviation measure. Nonetheless, this penalization term composed with a coherent (convex) risk measure $\rho$ results in a sub-additive (convex) loss-deviation. We now expose the formalization of such facts.
	
	\begin{Pro}\label{pro:LD}
		Let $\rho : L^{p}\rightarrow \mathbb{R}\cup\{\infty\}$ be a coherent (convex) risk measure and $LD^\beta_{\rho} : L^{p}\rightarrow \mathbb{R}\cup\{\infty\}$ its loss-deviation. Then:
		\begin{enumerate}
			\item $LD^\beta_{\rho}$ is a coherent (convex) risk measure. If $\rho$ is law invariant then $LD^\beta_{\rho}$ also is. Moreover, if $\rho$ is co-monotone, then $LD^\beta_{\rho}$ is sub-additive for any co-monotone pair $X,Y\in L^p$.
			\item If $\rho$ is Fatou continuous coherent, then $LD^\beta_{\rho}$ has, for $\mathcal{W}=\left\lbrace W:W\leq0,\lVert W\rVert_q\leq1  \right\rbrace$, dual set
			$\mathcal{P}_{LD^\beta_{\rho}}=\left\lbrace\mathbb{Q}\in\mathcal{P}:\frac{d\mathbb{Q}}{d\mathbb{P}}=\frac{d\mathbb{Q}_\rho}{d\mathbb{P}}(1+\beta E_\mathbb{P}[W])-\beta W ,\frac{d\mathbb{Q}_\rho}{d\mathbb{P}}\in\mathcal{P}_\rho,W\in\mathcal{W}  \right\rbrace$.
		\end{enumerate} 
	\end{Pro}
	
	\begin{proof}
		When $\beta=0$ all claims are obvious from the assumptions on $\rho$. We thus focus on the case $0<\beta\leq 1$. Regarding (i), Translation Invariance and Positive Homogeneity are easily obtained from the fact that $\rho$ fulfills such properties. Let $\Delta^{X+Y}=\lVert (X+Y-\rho^{*}(X+Y))^-\rVert_p-(\lVert (X-\rho^{*}(X))^-\rVert_p+\lVert (Y-\rho^{*}(Y))^-\rVert_p)$. From the Sub-additivity of $\rho$, we get for any $X,Y\in L^p$ that:
		\begin{align*}
			\Delta^{X+Y}&\leq\lVert (X+Y-\rho^{*}(X+Y))^-\rVert_p-\lVert (X+Y-\rho^{*}(X)-\rho^{*}(Y))^-\rVert_p\\
			&\leq\lVert(X+Y-\rho^{*}(X+Y))^--(X+Y-\rho^{*}(X)-\rho^{*}(Y))^-\rVert_p \\
			&\leq\lVert(X+Y-\rho^{*}(X+Y))^--(X+Y-\rho^{*}(X)-\rho^{*}(Y))^-\rVert_\infty \\
			&=\rho(X)+\rho(Y)-\rho(X+Y)\\
			&\leq\frac{1}{\beta}(\rho(X)+\rho(Y)-\rho(X+Y)).
		\end{align*}  
		Thus, we obtain $LD^\beta_{\rho}(X+Y)=\rho(X+Y)+\beta\lVert (X+Y-\rho^{*}(X+Y))^-\rVert_p\leq\rho(X)+\beta \lVert (X-\rho^{*}(X))^-\rVert_p+\rho(Y)+\beta \lVert (Y-\rho^{*}(Y))^-\rVert_p=LD^\beta_{\rho}(X)+LD^\beta_{\rho}(Y)$, as desired. The first and second inequalities are due to both the p-norm and negative part satisfying Sub-additivity property, while the last one is because $\rho$ fulfills Sub-additivity and $\beta\leq 1$. Moreover, since $\rho^{*}(X+Y)\geq\rho^{*}(X)+\rho^{*}(Y)$, we have that:
		$(X+Y-\rho^{*}(X+Y))^--(X+Y-\rho^{*}(X)-\rho^{*}(Y))^-$ assumes value $0$ if $X+Y\geq\rho^{*}(X+Y)$, $\rho^{*}(X+Y)-\rho^{*}(X)-\rho^{*}(Y)$ if $X+Y\leq\rho^{*}(X)+\rho^{*}(Y)$, and some scalar $C\leq\rho^{*}(X+Y)-\rho^{*}(X)-\rho^{*}(Y)$ otherwise. This explains the equality $\lVert(X+Y-\rho^{*}(X+Y))^--(X+Y-\rho^{*}(X)-\rho^{*}(Y))^-\rVert_\infty =\rho(X)+\rho(Y)-\rho(X+Y)$. When $\rho$ is a convex risk measure, the deduction is quite similar, beginning with $\Delta^{X+Y}_\lambda=\lVert (\lambda X+(1-\lambda)Y-\rho^{*}(\lambda X+(1-\lambda)Y))^-\rVert_p-(\lambda\lVert (X-\rho^{*}(X))^-\rVert_p+(1-\lambda
		)(\lVert (Y-\rho^{*}(Y))^-\rVert_p),\:0\leq\lambda\leq 1$. Thus, $LD^\beta_{\rho}$ fulfills Sub-additivity (Convexity) when $\rho$ does. This fact, together with the Limitedness axiom implies, from Proposition \ref{pro:lim}, that $LD^\beta_{\rho}$ satisfies Monotonicity.  Limitedness comes from the fact that $(X-\rho^{*}(X))^-\leq\rho^{*}(X)-\inf X,\:\forall X\in L^p$. Thus, we have $\rho^{*}(X)-\inf X\geq\lVert (X-\rho^{*}(X))^-\rVert_{\infty}\geq\lVert (X-\rho^{*}(X))^-\rVert_p\geq\beta\lVert (X-\rho^{*}(X))^-\rVert_p$. Hence, $LD^\beta_{\rho}(X)\leq-\inf X$. Hence $LD^\beta_{\rho}$ is a coherent or convex risk measure when $\rho$ lies in these same classes. Moreover, it is direct that the Law Invariance of $\rho$ implies the same axiom for $LD^\beta_{\rho}$ because the p-norm is based on expectation. Moreover, if $\rho$ is co-monotone, we have for co-monotonic $X,Y\in L^p$ that:
		\begin{align*}
			LD^\beta_{\rho}(X+Y)&=\rho(X+Y)+\beta\lVert (X+Y-\rho^{*}(X+Y))^-\rVert_p\\
			&=\rho(X)+\rho(Y)+\beta\lVert (X+Y-\rho^{*}(X)-\rho^{*}(Y))^-\rVert_p\\
			&\leq\rho(X)+\rho(Y)+\beta(\lVert (X-\rho^{*}(X))^-\rVert_p+\lVert (Y-\rho^{*}(Y))^-\rVert_p)\\
			&=LD^\beta_{\rho}(X)+LD^\beta_{\rho}(Y),
		\end{align*} 
		which is Sub-additivity for this case, as claimed.
		
		Concerning (ii), let $\mathcal{W}=\left\lbrace W:W\leq0,\lVert W\rVert_q\leq1  \right\rbrace$. It is well known, see \cite{Pflug2006}, that $\lVert X^-\rVert_p=\sup\limits_{W\in\mathcal{W}}E_\mathbb{P}[XW]$. From that, we can obtain for any $X\in L^p$ that:
		\begin{align*}
			LD^\beta_{\rho}(X)&=\rho(X)+\beta\sup\limits_{W\in\mathcal{W}}E_\mathbb{P}[(X-\rho^{*}(X))W]\\
			&=\beta\sup\limits_{W\in\mathcal{W}}\left\lbrace E_\mathbb{P}[XW]+\rho(X)\left(\frac{1}{\beta}+E_\mathbb{P}[W]\right)\right\rbrace\\
			&=\beta\sup\limits_{W\in\mathcal{W}}\left\lbrace E_\mathbb{P}[XW]+\left( \sup\limits_{\mathbb{Q}\in\mathcal{P}_\rho
			}E_{\mathbb{Q}}[-X]\right) \left(\frac{1}{\beta}+E_\mathbb{P}[W]\right)\right\rbrace\\
			&=\sup\limits_{\mathbb{Q}\in\mathcal{P}_\rho,W\in\mathcal{W}}\left\lbrace E_\mathbb{Q}[-X](1+\beta E_\mathbb{P}[W])+\beta E_{\mathbb{P}}[XW] \right\rbrace\\
			&=\sup\limits_{\mathbb{Q}\in\mathcal{P}_\rho,W\in\mathcal{W}}\left\lbrace E_{\mathbb{P}}\left[-X\left(\frac{d\mathbb{Q}_\rho}{d\mathbb{P}}(1+\beta E_\mathbb{P}[W])-\beta W  \right)  \right],\frac{d\mathbb{Q}_\rho}{d\mathbb{P}}\in\mathcal{P}_\rho\right\rbrace\\
			&=\sup\limits_{\mathbb{Q}\in\mathcal{P}_{LD^\beta_{\rho} }}E_{\mathbb{Q}}[-X],
		\end{align*}
		where
		\begin{equation*}
			\mathcal{P}_{LD^\beta_{\rho}}=\left\lbrace\mathbb{Q}\in\mathcal{P}:\frac{d\mathbb{Q}}{d\mathbb{P}}=\frac{d\mathbb{Q}_\rho}{d\mathbb{P}}(1+\beta E_\mathbb{P}[W])-\beta W ,\frac{d\mathbb{Q}_\rho}{d\mathbb{P}}\in\mathcal{P}_\rho,W\in\mathcal{W}  \right\rbrace.
		\end{equation*}
		In the third equality, we use the assumption that $\rho$ is Fatou continuous and, by Theorem \ref{the:dua}, possesses a dual representation. In this same equality, it is valid to note that $\beta^{-1}>1$ and $E_{\mathbb{P}}[W]\geq-1$, which implies that $\beta^{-1}+E_{\mathbb{P}}[W]\geq0$. It remains to show that $\mathcal{P}_{LD^\beta_{\rho} }$ is composed by valid probability measures, i.e. $\forall\:\mathbb{Q}\in\mathcal{P}_{LD^\beta_\rho}$ it is true that $\frac{d\mathbb{Q}}{d\mathbb{P}}\geq0$, $E_{\mathbb{P}}\left[\frac{d\mathbb{Q}}{d\mathbb{P}} \right]=1 $, and $ \frac{d\mathbb{Q}}{d\mathbb{P}}\in L^q$. In this sense, since $\beta E_\mathbb{P}[W]\geq E_\mathbb{P}[W]\geq -1$, we get that $\frac{d\mathbb{Q}}{d\mathbb{P}}=\frac{d\mathbb{Q}_\rho}{d\mathbb{P}}(1+\beta E_\mathbb{P}[W])-\beta W\geq0,\:\forall\:\mathbb{Q}\in\mathcal{P}_{LD^\beta_{\rho}}$. Moreover, we have that $E_\mathbb{P}\left[\frac{d\mathbb{Q}}{d\mathbb{P}}\right]=E_\mathbb{P}\left[\frac{d\mathbb{Q}_\rho}{d\mathbb{P}}(1+\beta E_\mathbb{P}[W])-\beta W\right]=E_\mathbb{P}\left[\frac{d\mathbb{Q}_\rho}{d\mathbb{P}}\right](1+\beta E_\mathbb{P}[W])-\beta E_\mathbb{P}[W]=1,\:\forall\:\mathbb{Q}\in\mathcal{P}_{LD^\beta_{\rho}}$. Finally, we also have that $\left\| \frac{d\mathbb{Q}}{d\mathbb{P}}\right\| _q=\left\| \frac{d\mathbb{Q}_\rho}{d\mathbb{P}}(1+\beta E_\mathbb{P}[W])-\beta W\right\|_q\leq(1+\beta E_\mathbb{P}[W])\left\| \frac{d\mathbb{Q}_\rho}{d\mathbb{P}}\right\|_q+\beta\left\| W\right\|_q<\infty$. This concludes the proof.
	\end{proof}
	
	\begin{Rem}
		When $\rho(X)=-E[X]$, we obtain the mean plus semi-deviation as a particular case, where relative densities have the form $\frac{d\mathbb{Q}}{d\mathbb{P}}=\frac{d\mathbb{Q}_\rho}{d\mathbb{P}}(1+\beta E_\mathbb{P}[W])-\beta W=1+\beta(E_{\mathbb{P}}[W]-W),\:W\in\mathcal{W}$, the same as in Proposition \ref{Pro:MSD}. 
	\end{Rem}
	
	Despite the fact that this kind of risk measure is not contemplated by our main Theorems \ref{the:main}, \ref{the:con} and \ref{the:com}, its coherence (convexity) is guaranteed by Proposition \ref{pro:lim}. This reinforces the role of Limitedness when one combines risk and deviation measures. If $\rho$ fulfills Law Invariance and Co-monotonic Additivity, then $LD^\beta_{\rho}$ lies in the more flexible class of Natural risk measures proposed by \cite{Kou2013}, which must satisfy Monotonicity, Translation Invariance, Positive Homogeneity, Law Invariance and Co-monotonic Sub-Additivity. There are other examples in the literature of functionals composed by a coherent risk measure and a non-convex deviation that is again a coherent risk measure. This is exactly what happens for the Tail Gini Shortfall, proposed by \cite{Furman2017}, and its extension introduced in \cite{Berkhouch2017}. The idea of such risk measures is to have a composition of the form $\rho+\beta\mathcal{D}$ between ES and a Gini functional restricted to the distribution tail. In both cases, it is a necessary restriction on the range of values for $\beta$. In general, it is possible to 'force' Limitedness over $\rho+\mathcal{D}$ by replacing $\mathcal{D}$ for $\beta\mathcal{D}$, under some restriction on the range of $\beta$, despite the properties of $\mathcal{D}$. We now provide such results in a formal way.
	
	\begin{Pro}\label{Pro:bet}
		Let $\rho : L^{p}\rightarrow \mathbb{R}\cup\{\infty\}$ be limited and $\mathcal{D} : L^{p}\rightarrow \mathbb{R}_{+}\cup\{\infty\}$. Then, $\rho+\beta\mathcal{D}$ fulfills Limitedness if and only if $\beta\leq\inf \left\lbrace \frac{\rho^{*}(X)-\inf X }{\mathcal{D}(X)}:X\in L^p,\mathcal{D}(X)>0\right\rbrace$.
	\end{Pro}
	
	\begin{proof}
		When $\mathcal{D}(X)=0$, it is straightforward that Limitedness is achieved by the assumption on $\rho$. Thus, we focus on the cases when $\mathcal{D}(X)>0$ composed by those non-constant $X$. Let $K=\inf \left\lbrace \frac{\rho^{*}(X)-\inf X }{\mathcal{D}(X)}:X\in L^p,\mathcal{D}(X)>0\right\rbrace$. For $\beta\leq K$, we thus get the following:
		\begin{align*}
			\rho(X)+\beta\mathcal{D}(X)&\leq\rho(X)+K\mathcal{D}(X)\\
			&\leq\rho(X)+\left(\frac{\rho^{*}(X)-\inf X }{\mathcal{D}(X)} \right) \mathcal{D}(X)\\
			&=-\inf X.
		\end{align*}
		For the converse relation, we now assume that $\rho(X)+\beta\mathcal{D}(X)\leq-\inf X,\:\forall\:X\in L^p$. In this case, we obtain that:
		\begin{align*}
			\beta&\leq\inf \left\lbrace \frac{\rho^{*}(X)-\inf X }{\mathcal{D}(X)}:X\in L^p\right\rbrace\\
			&\leq\inf \left\lbrace \frac{\rho^{*}(X)-\inf X }{\mathcal{D}(X)}:X\in L^p,\mathcal{D}(X)>0\right\rbrace.
		\end{align*}
		This concludes the proof.	
	\end{proof}
	
	\begin{Rem}
		In the proposition, we have not used the typical practical constraint $\beta\geq 0$ because, in this case, we have directly achieved the result since $\rho+\beta\mathcal{D}\leq\rho$. Due to Proposition \ref{pro:lim}, we have that $\rho$ is limited for all frameworks of Theorems \ref{the:main}, \ref{the:con} and \ref{the:com}. Thus, this last result is not restricted. Regarding practical interpretation, if $\rho^{*}$ is a parsimonious risk measure that is far from $\inf X$, then $\beta$ can assume larger values, i.e. more protection from $\mathcal{D}$ can be added without losing Limitedness. The contrary reasoning is also valid.
	\end{Rem}
	
	Even with this constraint for the value of $\beta$, Monotonicity is not necessarily achieved because, under Limitedness, Sub-Additivity or Convexity is demanded by Proposition \ref{pro:lim}. At the same time, it would be nice to obtain the deviation term in our composition directly from the risk measure we choose, as in the loss-deviation approach. Under all these perspectives, we consider, as our next example, deviations defined as risk measures applied over a demeaned financial position. We now define such types of deviation.
	
	\begin{Def}
		Let $\rho:L^p\rightarrow\mathbb{R}\cup\{\infty\}$ be a risk measure. Then the deviation induced by $\rho$ is a functional $\mathcal{D}^\rho:L^p\rightarrow\mathbb{R}\cup\{\infty\}$ defined conform:
		\begin{equation*}
			\mathcal{D}^\rho(X)=\rho(X-E_{\mathbb{P}}[X]).
		\end{equation*}
	\end{Def}
	
	This characterization is explored in \cite{Rockafellar2006} and \cite{Rockafellar2013}, which prove that $\mathcal{D}^\rho$ is indeed a lower range dominated generalized (convex) deviation when $\rho$ is a coherent (convex) risk measure. For the case of coherent $\rho$ it is not hard to realize that $\mathcal{P}_{\mathcal{D}^\rho}=\mathcal{P}_\rho$. It is interesting to consider risk measures strictly larger than the negative expectation in order to have a well-defined situation since $E_\mathbb{P}[X-E_\mathbb{P}[X]]=0,\:\forall\:X\in L^p$. We now investigate the properties of a composition given by $\rho+\beta\mathcal{D}^\rho$.
	
	\begin{Pro}\label{Pro:DRho}
		Let $\rho : L^{p}\rightarrow \mathbb{R}\cup\{\infty\}$ be a coherent (convex) risk measure such as $\rho(X)>-E_\mathbb{P}[X],\:\forall\:X\in L^p$. Then:
		\begin{enumerate}
			\item The composition $\rho+\beta\mathcal{D}^\rho$ is a coherent (convex) risk measure if and only if $0\leq\beta\leq\inf \left\lbrace \frac{\rho^{*}(X)-\inf X }{E_\mathbb{P}(X)-\rho^{*}(X)}:X\in L^p\right\rbrace$. Moreover, if $\rho$ fulfills Law Invariance (Co-monotonic Additivity), then $\rho+\beta\mathcal{D}^\rho$ is law invariant (co-monotone).
			\item If $\rho$ is Fatou continuous coherent and $0\leq\beta\leq\inf \left\lbrace \frac{\rho^{*}(X)-\inf X }{E_\mathbb{P}(X)-\rho^{*}(X)}:X\in L^p\right\rbrace$, then the combination $\rho+\beta\mathcal{D}^\rho$ has dual set $\mathcal{P}_{\rho+\beta\mathcal{D}^\rho}=\left\lbrace\mathbb{Q}\in\mathcal{P}:\frac{d\mathbb{Q}}{d\mathbb{P}}=(1+\beta)\frac{d\mathbb{Q}^\rho}{d\mathbb{P}}-\beta ,\mathbb{Q}^\rho\in\mathcal{P}^\rho\right\rbrace $.
		\end{enumerate}
	\end{Pro} 
	
	\begin{proof}
		For (i), Translation Invariance, Sub-Additivity, Convexity, Positive Homogeneity, Law Invariance and Co-monotonic Additivity are direct when $\rho$ possess these properties and $\beta\geq0$. Regarding Monotonicity, from Proposition \ref{pro:lim} we just need Limitedness because Sub-additivity (Convexity) is present. Moreover, by Proposition \ref{Pro:bet} we have that $\rho+\beta\mathcal{D}^\rho$ is limited if and only if $\beta\leq\inf \left\lbrace \frac{\rho^{*}(X)-\inf X }{\rho(X-E_\mathbb{P}[X])}:X\in L^p,\rho(X-E_\mathbb{P}[X])>0\right\rbrace=\left\lbrace \frac{\rho^{*}(X)-\inf X }{E_\mathbb{P}[X]-\rho^*(X)}:X\in L^p\right\rbrace$, as claimed.
		
		Concerning (ii), since $\rho$ is Fatou continuous coherent, it has a representation under dual set $\mathcal{P}_\rho$. The same is true for $\beta\mathcal{D}^\rho$ under relative densities $\beta\frac{d\mathbb{Q}_\rho}{d\mathbb{P}}-(1-\beta)$. The role of $\beta$ is according to \cite{Rockafellar2006}. Moreover, from (i), $\rho+\beta\mathcal{D}^\rho$ is also coherent for $0\leq\beta\leq\inf \left\lbrace \frac{\rho^{*}(X)-\inf X }{E_\mathbb{P}(X)-\rho^{*}(X)}:X\in L^p\right\rbrace$. Its Fatou continuity is direct from that of $\rho$. From Theorem \ref{the:main} its dual set is $\mathcal{P}_{\rho+\beta\mathcal{D}^\rho}=\left\lbrace\mathbb{Q}\in\mathcal{P}:\frac{d\mathbb{Q}}{d\mathbb{P}}=\frac{d\mathbb{Q}_\rho}{d\mathbb{P}} +\beta\frac{d\mathbb{Q}_\rho}{d\mathbb{P}}-(1-\beta)-1 ,\mathbb{Q}_\rho\in\mathcal{P}_\rho\right\rbrace $. After some simple manipulation, the claim is achieved. This concludes the proof.
	\end{proof}
	
	\section{Conclusion}\label{sec:conc}
	In this paper, we present a composition of risk and deviation measures, which considers the concepts of loss and variability, in order to keep desired theoretical properties. Most studies are so far only concerned with specific examples, while we present a general approach. Our results are based on the proposed Limitedness axiom, which indicates that the composition value must not be over a certain limit -- the supremum of possible losses. In this context, we prove that this composition is a coherent, convex or co-monotone risk measure, conforming to properties of the two components. In a second contribution, we provide results about specific examples of known and new risk measures constructed under this framework. In such results, the importance of our approach becomes clear, especially the role of Limitedness axiom. 
	
	\bibliographystyle{elsarticle-harv}
	\bibliography{rep}
\end{document}